\pdfoutput=1

\documentclass[11pt,a4paper]{article}

\usepackage{amsmath,amssymb,amsthm} 
\usepackage{mathptmx}

\usepackage{fullpage}

\usepackage{graphicx}

\newtheorem{theorem}{Theorem}[section]
\newtheorem{lemma}[theorem]{Lemma}

\newtheorem*{claim}{Claim}
 
\theoremstyle{definition}

\newcommand{\Class}[1]{\operatorname{\mathchoice{\text{\small
        $\mathrm{#1}$}} {\text{\small
        $\mathrm{#1}$}}{\mathrm{#1}}{\mathrm{#1}}}}
\newcommand{\Lang}[1]{{\leavevmode\hbox{\textsc{#1}}}}
\newenvironment{Enumerate}{\enumerate\itemsep=0pt\parskip=0pt}{\endenumerate}
\let\origthebibliography=\thebibliography
\renewcommand\thebibliography[1]{\small\origthebibliography{#1}\parskip0pt\itemsep0pt}

\usepackage{url}

\begin{document}

\title{Perfect Phylogeny Haplotyping is Complete for Logspace}
\author{Michael Elberfeld\\ \small Institut f\"ur Theoretische
  Informatik\\[-2pt] \small Universit\"at zu L\"ubeck\\[-2pt] \small
  D-23538 L\"ubeck, Germany\\[-2pt] \small
  \path|elberfeld@tcs.uni-luebeck.de| }

\date{
\vspace*{0.5cm} \today}

\maketitle

\begin{abstract}
  Haplotyping is the bioinformatics problem of predicting likely haplotypes based on
  given genotypes. It can be approached using Gusfield's perfect
  phylogeny haplotyping (\Lang{pph}) method for which polynomial and
  linear time algorithms exist. These algorithm use sophisticated data
  structures or do a stepwise transformation of the genotype data into haplotype
  data and, therefore, need a linear amount of space. We are interested
  in the exact computational complexity of \Lang{pph} and show that it
  can be solved space-efficiently by an algorithm that needs only a
  logarithmic amount of space. Together with the recently proved
  $\Class{L}$-hardness of \Lang{pph}, we establish
  $\Class{L}$-completeness. Our algorithm relies on a new
  characterization for \Lang{pph} in terms of bipartite graphs, which
  can be used both to decide and construct perfect phylogenies for
  genotypes efficiently.
\end{abstract} 

\section{Introduction}
In human genetic variation studies, sequencing methods are applied
that read out the genetic information at \textsc{snp} (single
nucleotide polymorphism) sites for multiple individuals. In order to
be low-priced and feasible, these methods determine, for each site
separately, the present bases, of which there can be two since the
human \textsc{dna} is arranged in pairs of chromosomes. For each
individual in the variation study this yields a \emph{genotype} that
describes the bases at \textsc{snp} sites.
While the genotype says for every site which bases are present, it
lacks the information on how the bases are assigned to the chromosomes
of a pair. This information, which is described by \emph{haplotypes},
is crucial to describe fine-grained genetic variation.

The objective of \emph{haplotyping} is to compensate the drawback of
genotype data by predicting biologically reasonable
haplotypes computationally. Gusfield~\cite{Gusfield2002} proposed an approach to
haplotyping that seeks haplotypes that are arrangeable in a
perfect phylogenetic tree~\cite{Gusfield2002}. He showed that this
problem, which will be called \emph{perfect phylogeny haplotyping} (\Lang{pph}) is
solvable in polynomial time by a reduction to the graph realization
problem .
Due to the practical importance of haplotyping, several groups also
proposed simpler polynomial time~\cite{Bafnaetal2003,EskinHK2003} and
linear time algorithms~\cite{DingFG2006,LiuZ2005,SatyaM2006} for
Gusfield's approach.

In the present paper we study the space complexity of \Lang{pph}.
In \cite{ElberfeldT2008b} we showed that 
\Lang{pph} is hard for the complexity class $\Class{L}$ (deterministic
logarithmic space) and lies in the counting class 
$\Class{\oplus L}$~\cite{ElberfeldT2008b} (see this paper for a wider
discussion of the
haplotyping issue and complexity theoretic terms). The main open
problem of~\cite{ElberfeldT2008b}, namely, whether \Lang{pph} lies in the class $\Class{L}$, is
answered affirmatively by the present paper. To prove this result, we
present a graph-based characterization that extends ideas from Eskin,
Halperin and Karp~\cite{EskinHK2003}. Given a set of genotypes, they
build, for each genotype separately, graphs where the vertices
represent sites and edges represent known relations between them.
Based on these graphs, they proved that the existence of a perfect
phylogeny is related to the question
whether one can extend the known relations, such that all graphs
become complete bipartite. Our characterization avoids the step of
guessing new relations between pairs of sites: We determine all
relevant relations beforehand and directly construct graphs that are bipartite
if, and only if, there is a perfect phylogeny. Since the graph
construction can be described by first-order formulas and the problem of 
deciding whether a graph is bipartite lies in 
$\Class{L}$~\cite{Reingold2008}, we are able to prove the following  
theorem: 
\begin{theorem}\label{theorem:completeness}
  \Lang{pph} is complete for deterministic logarithmic space.
\end{theorem}

The paper is organized as follows: In Section~\ref{section:basics} we
provide a formal definition of the \Lang{pph} problem and induced sets
of pairs of sites. In Section~\ref{section:characterization}, we
prove our graph based-characterization and, after that, we
reduce~\Lang{pph} to the
problem of whether an undirected graph is
bipartite in Section~\ref{section:reduction}.

\section{Perfect Phylogenies and Induced Sets}\label{section:basics}

Since only two different bases are present at the majority of
\textsc{snp} sites, it is convenient to code haplotypes as strings
over the alphabet $\{0,1\}$, where for a given site $0$ stands for one
of the bases that can be observed in practice, while $1$ encodes a
second base that can also be observed. A genotype $g$ is a sequence of
sets that arises from a pair of haplotypes $h$ and $h'$ as follows:
The~$i$th set in the sequence $g$ is $\{h[i],h'[i]\}$. However, it is
customary to encode the set $\{0\}$ as $0$, to encode $\{1\}$ as $1$,
and $\{0,1\}$ as $2$, so that a genotype is actually a string over the
alphabet $\{0,1,2\}$. For example, the two haplotypes $0100$ and
$0111$ underly (we also say \emph{explain}) the genotype $0122$; and
so do $0101$ and $0110$. These haplotype pairs differ
in the way how the 2-entries at positions three and four are
determined. The haplotypes have the same entries at positions three and
four in the first case and different entries in the second case.
This fact can be generally stated as follows: If $h$ and $h'$
are explaining haplotypes for a genotype $g$ with 2-entries in
sites~$i$ and~$j$ ($g[i] = g[j] = 2$), then either $h[i] = h[j] \neq h'[i]
= h'[j]$ or $h[i] = h'[j] \neq h'[i] = h[j]$ holds.  In the first case
we say that \emph{$h$ and $h'$ resolve $g$ equally in~$i$ and~$j$} and
in the second case we say that \emph{$h$ and $h'$ resolve $g$
  unequally in~$i$ and~$j$}. To represent more than one haplotype, we
arrange them in \emph{haplotype matrices} where each row is a
haplotype and each column corresponds to a site. For genotypes, we use
\emph{genotype matrices}. An $2n \times m$ haplotype matrix $B$
\emph{explains} an $n \times m$ genotype matrix $A$ if for each~$i$,
the haplotypes in rows $2i-1$ and $2i$ of $B$ explain the genotype in
row~$i$ of $A$. 

We are interested in haplotypes that are arrangeable in a perfect
phylogenetic tree. We say that a haplotype matrix $B$ \emph{admits a
  perfect phylogeny} if there exists a rooted tree $T$, such that:
\begin{Enumerate}
\item Each row of $B$ labels exactly one node of $T$.
\item Each column of $B$ labels exactly one edge of $T$ and each edge
  is labeled by at least one column.
\item For every two rows $h$ and $h'$ of $B$ and every column~$i$, we
  have $h[i] \neq h'[i]$ if, and only if, $i$ lies on the path from $h$ to $h'$ in
  $T$.
\end{Enumerate}
A haplotype matrix $B$ \emph{admits a directed perfect phylogeny} if
$B$ together with the all-0-haplotype admits a perfect phylogeny.  The
\emph{four gamete property} is an alternative characterization for
perfect phylogenies, observed by many authors
(see~\cite{GrammNT2007a} for references). It depends on a certain
relation between pairs of columns: The \emph{induced set} $\operatorname{ind}^B(i,j)$
of two columns~$i$ and~$j$ in a haplotype matrix $B$ contains all
strings from $\{00,01,10,11\}$ that appear in the columns~$i$
and~$j$. The four gamete property then says that a haplotype matrix $B$ admits a perfect
phylogeny if, and only if, for each pair of columns~$i$ and~$j$ 
we have $\{00,01,10,11\} \neq \operatorname{ind}^B(i,j)$. Carried over to the
directed case we know that $B$ 
admits a directed perfect phylogeny if, and only if, for each pair of
columns~$i$ and~$j$ we have $\{01,10,11\} \nsubseteq
\operatorname{ind}^B(i,j)$. We refer to this as the \emph{three gamete
  property} in the following.

We say that a genotype matrix $A$ \emph{admits a (directed) perfect
  phylogeny} if there exists an explaining haplotype matrix for it that
admits a (directed) perfect phylogeny or, equivalently, satisfies the
four (three) gamete property. The perfect phylogeny haplotyping problem
($\Lang{pph}$) contains exactly the genotype matrices that admit a
perfect phylogeny. Similar, the directed perfect phylogeny haplotyping
problem
($\Lang{dpph}$) contains exactly the genotype matrices that admit a
directed perfect phylogeny. The problems \Lang{pph} and \Lang{dpph}
are closely related through first-order reductions: For a
reduction from \Lang{dpph} to \Lang{pph}
is suffices to append the all-0-genotype to a given genotype
matrix and for the converse direction we can use a reduction from 
Eskin, Halperin and Karp~\cite{EskinHK2003}: In every column where a
1-entry appears before a 0-entry, substitute all 1-entries by
0-entries and all 0-entries by 1-entries. For convenience we restrict
ourselves 
to directed perfect phylogenies in the rest if this section and
Section~\ref{section:characterization}. We come back to undirected
perfect phylogenies in Section~\ref{section:reduction}.

A genotype matrix determines, to a certain extend, the induced sets of
explaining haplotype matrices. This is
formalized by the notion of induced sets for genotype
matrices in~\cite{EskinHK2003}: For a genotype matrix $A$ and two columns~$i$ and~$j$, the
set $\operatorname{ind}^A(i,j)$ contains a string $xy \in
\{00,01,10,11\}$, whenever $A$ has a genotype $g$ with
either $g[i] = x$ and $g[j] = y$, $g[i] = x$ and $g[j] = 2$ or $g[i] =
2$ and $g[j] = y$. From this definition follows that we have
$\operatorname{ind}^A(i,j) \subseteq \operatorname{ind}^B(i,j)$ for
any haplotype matrix $B$ explaining $A$ and $\operatorname{ind}^A(i,j)
= \operatorname{ind}^B(i,j)$ if $A$ does not contain a genotype with
2-entries in both $i$ and $j$. Also we know that 
$A$ does not admit a directed perfect phylogeny whenever
$\{01,10,11\} \subseteq \operatorname{ind}^A(i,j)$ holds.

If we consider only explaining haplotype matrices that satisfy the three
gamete property, we can infer some information about the resolution of
2-entries: Let $A$ be a genotype matrix and $B$ an explaining haplotype
matrix for $A$ that satisfies the three gamete property. Whenever we
have $\{01,10\} \subseteq \operatorname{ind}^A(i,j)$ for
columns~$i$ and~$j$, we know that every genotype $g$ of $A$ with $g[i]
= g[j] =2$ is resolved unequally in~$i$ and~$j$ by its haplotypes from
$B$. Whenever we have $\{11\} \subseteq \operatorname{ind}^A(i,j)$, the
genotypes are resolved equally in~$i$ and $j$. We can also infer
resolutions through genotypes with at least three 2-entries:
Consider a genotype $g$ from $A$ and columns $i$, $j$ and $k$ with 
$g[i] = g[j] = g[k] = 2$ and induced sets $\{11\} \subseteq
\operatorname{ind}^A(i,j)$ and $\{01,10\} \subseteq
\operatorname{ind}^A(j,k)$. Let $h$ and $h'$ be the explaining
haplotypes for $g$ from $B$. We know from the induces that $h[i]
= h[j] \neq h'[i] = h'[j]$ and $h[j] = h'[k] \neq h'[j] = h[k]$ hold.
This implies $h[i] = h'[k] \neq h'[i] = h[k]$, and, therefore, every
genotype with 2-entries
in~$i$ and~$k$ must be resolved equally by its haplotypes in that
columns. Note, that we do not know the resolution in columns $i$ and
$k$ from the induced set of these column pair. We deduced the
resolution through three 2-entries in $g$ by using information about the
induced sets $\operatorname{ind}^A(i,j)$ and $\operatorname{ind}^A(j,k)$.
The derived equal resolution in columns~$i$ and $k$ may, again,
trigger a resolution in another pair of columns $k$ and~$l$
through a genotype with 2-entries in~$i$, $k$ and~$l$. In this way
resolutions may propagate through column pairs of the whole genotype
matrix. This can possibly end up with a column pair where one genotype
is forced to be resolved equally, while another is already resolved
unequally. In this case $A$ does not admit a perfect phylogeny. In the
next section we describe graphs that represent resolutions in column
pairs and their propagation through the genotype matrix. 

\section{A Graph-Based Characterization for Perfect Phylogeny Haplotyping} \label{section:characterization}

We present a new characterization for \Lang{dpph} in terms of
undirected edge-weighted graphs. The graphs are used to represent
resolutions in column pairs and their propagation in the
genotype matrix. Therefore we call these graphs \emph{resolution graphs}. The
vertices of each resolution graph are identified with columns of a
given genotype matrix and the edges are weighted by $0$ or $1$. An
edge with weight $0$ between two vertices $k$ and~$l$ indicates that
all genotypes must be resolved equally in columns $k$ and~$l$. Similar,
an edge with weight $1$ indicates that the columns are resolved
unequally. Our characterization in Lemma~\ref{lemma:characterization}
says that the
absence of odd-weight cycles (the weight of a path or a cycle is the
sum of its edge weights) from the resolution graphs is equivalent to
the fact that there is a directed perfect phylogeny. 

Given an $n \times m$ genotype matrix $A$, we build $m$ resolution
graphs, for every column~$i$ a single graph~$G_i$. A graph~$G_i$
describes resolutions for a particular set $A_i$ of genotypes from
$A$ and we define the $A_i$ sets such that they are all pairwise
disjoint. If one wants to determine the haplotypes for a
particular genotype $g$, it will suffice to consider the unique graph~$G_i$
with $g \in A_i$. To assign the genotypes to the sets $A_i$, we use 
a partial order on genotype matrix columns
from~\cite{EskinHK2003}: A column
with index~$i$ \emph{is greater than} a column with index~$j$ (denoted
by $i \succ^A j$) if $\operatorname{ind}^A(i,j) \subseteq
\{00,10,11\}$ and the column vectors are not the
same (see Figure~\ref{figure:characterizationexample} for an example). Beside
this partial order, we also
use the total order that is given by the indices of the columns of $A$. For
every $i \in
\{1,\ldots,m\}$, the set $A_i$ is then defined as follows:
\begin{align*}
A_i = \{ & \text{genotype}\ g\ \text{of}\ A \mid\ g[i] = 2,\\ 
&~~\forall j \neq i\ (g[j] = 2 \Longrightarrow\ \text{not}\ j \succ^A
i)\ \text{and}\\ 
&~~\forall j \neq i\ ((g[j] = 2\ \text{and}\ \forall k \neq j\ (g[k] = 2
\Longrightarrow\ \text{not}\ k \succ^A j)) \Longrightarrow\ j > i)\}\ .
\end{align*}
This definition assigns every genotype with a 2-entry to exactly one
set $A_i$. Genotypes without 2-entries 
do not need any resolution and, therefore, they are not assigned to any set. 
Figure~\ref{figure:characterizationexample} shows an 
example of a genotype matrix $A$ and its sets $A_i$.

For every $i \in \{1,\dots,m\}$, we now define the resolution graph
$G_i$ (again, see Figure~\ref{figure:characterizationexample} for an
example). As already stated, the vertices $V_i$ of $G_i$ are identified with
columns from $A$. A vertex with index $k$ lies in $V_i$ if, and only if, $A_i$
contains a genotype with a 2-entry in column $k$. The edges 
$E_i \subseteq \{e \subseteq V_i \mid |e| = 2\}$
and their weights $w_i: E_i \to \{0,1\}$ are constructed as follows:
\begin{align*}
  \{k,l\} \in E_i\ \text{and}\ w_i(\{k,l\}) = 0\ & \text{if, and only if, there exists}\ g_1 \in
  A_i\ \text{with}\ g_1[k] = g_1[l] = 2\ \text{and}\\ 
  & \text{(a)}\ 11 \in \operatorname{ind}(k,l),\ \text{or}\\ 
  & \text{(b)}\ \text{there is a column}\ j \neq i\ \text{and}\ g_2 \in
  A_j\ \text{with}\ g_2[k] = g_2[l] = 2\\[1.5ex] 
  \{k,l\} \in E_i\ \text{and}\ w_i(\{k,l\}) = 1\ & \text{if, and only if, there exists}\ g_1 \in A_i\ \text{with}\ g_1[k] = g_1[l] = 2\\ 
  & \text{and}\ \{01,10\} \subseteq \operatorname{ind}^A(k,l)
  \end{align*}

\begin{figure} [htpb]
\begin{center}
  \includegraphics{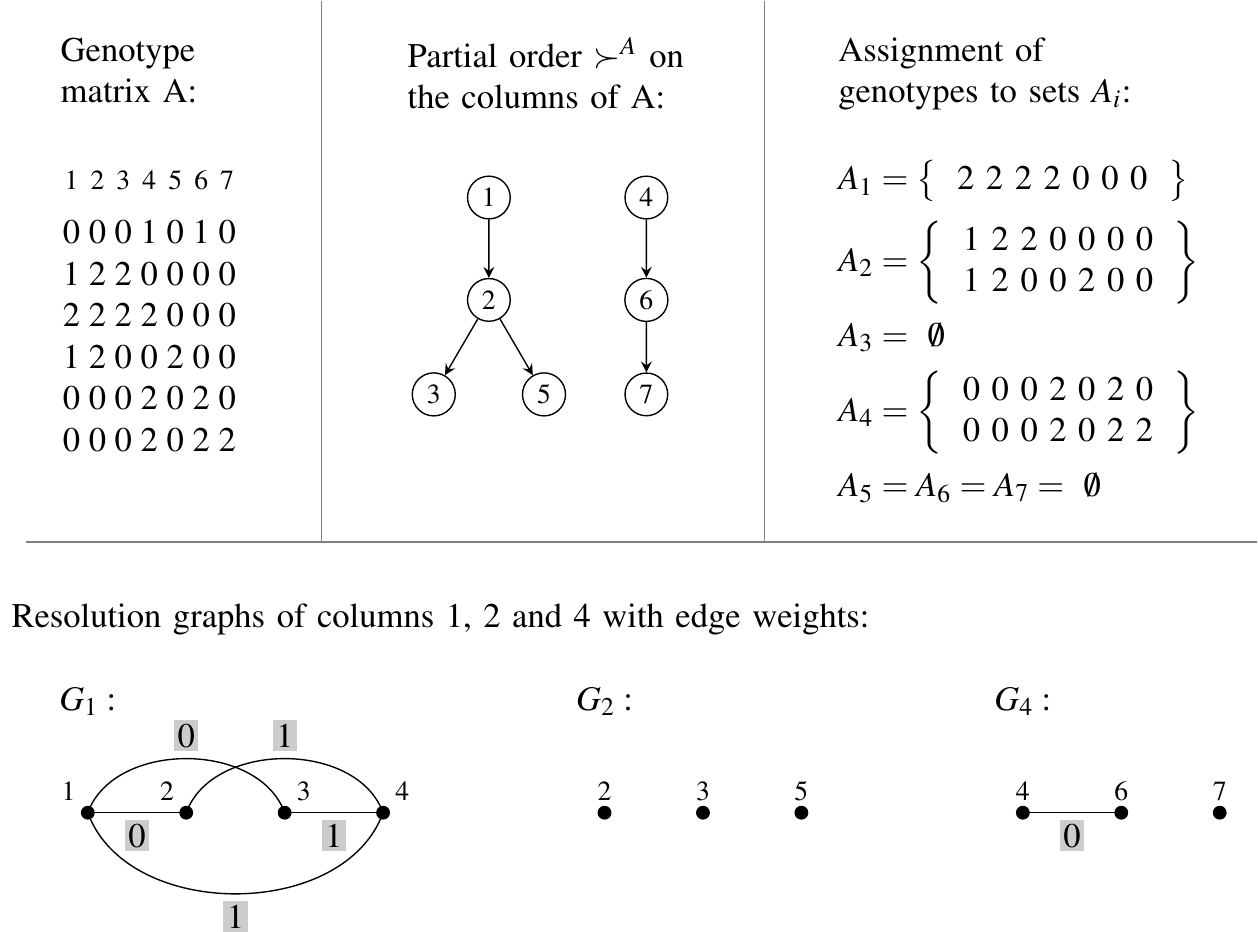}
\end{center}
     \caption[]{
       For a genotype matrix $A$ with seven columns, this figure shows
       the corresponding partial order~$\succ^A$ on the columns of
       $A$, the assignment of genotypes to sets $A_i$, and the
       construction of resolution graphs $G_i$. Some sets $A_i$ are
       empty and, therefore, the corresponding resolution graphs
       are also empty. A directed edge in the partial order from a vertex with
       index $k$ to a vertex with index $l$ means $k \succ^A l$.}
     \label{figure:characterizationexample}
\end{figure}

The characterization for \Lang{dpph} is as follows:
\begin{lemma}\label{lemma:characterization}
  An $n \times m$ genotype matrix $A$ admits a directed perfect
  phylogeny if, and only if, for each pair $i,j \in \{1,\dots,m\}$, we
  have $\{01,10,11\} \nsubseteq \operatorname{ind}^A(i,j)$, and for each $i \in
  \{1,\dots,m\}$,~$G_i$ does not contain an odd-weight cycle.
\end{lemma}
\begin{proof}
  \emph{Only-if-part:} Let $A$ be a genotype matrix and $B$ a
  haplotype matrix for it that satisfies the three gamete property. Thus, for every 
  pair of columns $i$ and $j$, we have $\{01,10,11\} \nsubseteq
  \operatorname{ind}^B(i,j)$ and, therefore, $\{01,10,11\} \nsubseteq
  \operatorname{ind}^A(i,j)$. To prove that none of the resolution graphs has an odd-weight cycle,
  we first show the following property: 
  \begin{claim}
    Let $A$ be a genotype matrix and $B$ a haplotype matrix for it that
    satisfies the three gamete property. Let $i$, $k$ and $l$ be columns
    and $g$ a genotype with $g\in A_i$ and $g[k] = g[l] = 2$. If $G_i$
    contains an edge with weight $1$ between $k$ and $l$, then $B$ resolves
    $g$ unequally in $k$ and $l$. If the weight is $0$, then the
    resolution is equal.
  \end{claim}
  \begin{proof}
    If there is a $1$-weighted edge between $k$ and $l$, we have $\{01,10\}
    \in \operatorname{ind}^A(k,l)$ and, therefore, $g$ must be resolved
    unequally in $k$ and $l$. If there is a $0$-weighted edge between $k$ and
    $l$, it is constructed for one of two reason: Whenever $11 \in
    \operatorname{ind}^A(k,l)$, we know by the three gamete property that
    $g$ must be resolved equally in $k$ and $l$. We are left with the case
    that there is another column~$j$ and a genotype $g_2 \in A_j$ with $g_2[k]
    = g_2[l] = 2$. In this case, the following matrix shows what we know about
    the entries of~$g_1$ and $g_2$, where $g_1[j]$ and $g_2[i]$ are values from $\{0,1,2\}$:
  \begin{center}
  \includegraphics{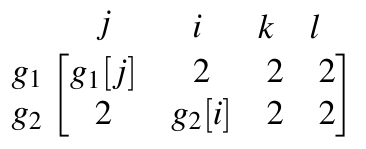}
  \end{center}
  By definition a genotype is contained in a set $A_i$ if column~$i$ is the
  maximal column (with respect to $\succ^A$) with the lowest index among all
  columns with a 2-entry in the genotype. Since $i$ and $j$ are not the same,
  this implies that at least one of the entries $g_1[j]$ and $g_2[i]$ does not
  equal 2.  We distinguish between the possible values for them (possible
  values are $g_2[i] = 1$, $g_2[i]=0$, $g_1[j] = 1$ and $g_1[j] =0$) and show
  that we have $h_1[k] = h_1[l] \neq h_1'[k] = h_1'[l]$ for the explaining
  haplotypes $h_1$ and $h_1'$ of $g_1$:

  Case $g_2[i] = 1$: We know $11 \in \operatorname{ind}^A(i,k)$, $11
      \in \operatorname{ind}^A(i,l)$ and, therefore, $h_1[i] = h_1[k] \neq
      h_1'[i] = h_1'[k]$ and $h_1[i] = h_1[l] \neq h_1'[i] = h_1'[l]$ which
      implies $h_1[k] = h_1[l] \neq h_1'[k] = h_1'[l]$.

  Case $g_2[i] = 0$: By assumption column~$k$ is not greater than
    column~$i$ which implies, together with the fact that they are not the
    same, $10 \in \operatorname{ind}^A(i,k)$.  Similarly, we have $10 \in
    \operatorname{ind}^A(i,l)$ for columns~$i$ and~$l$. Furthermore, the
    $0$-entry in $g_2[i]$ and the 2-entries in~$g_2[k]$ and~$g_2[l]$ ensure
    that $01 \in \operatorname{ind}^A(i,k)$ and $01 \in
    \operatorname{ind}^A(i,l)$. The resulting induced sets force an unequally
    resolution of~$g_1$ in both~$i$ and~$k$, and~$i$ and~$l$. Thus $h_1[i] =
    h_1'[k] \neq h_1'[i] = h_1[k]$, $h_1[i] = h_1'[l] \neq h_1'[i] = h_1[l]$
    and, therefore, $h_1[k] = h_1[l] \neq h_1'[k] = h_1'[l]$.

  Case $g_1[j] = 1$ and $g_1[j] = 0$ are similar to case $g_2[i] = 1$
  and case $g_2[i] = 0$, respectively. Thus, we proved the claim.
  \end{proof}
  
  We assume, for sake of contradiction, that there is a column~$i$
  such that~$G_i$ contains an odd-weight cycle $i_1, e_{1},i_2 \dots i_p, e_{p},
  i_{p+1} = i_1$, where the $i_j$ are column
  indices and the $e_{j}$ are weighted edges from $E_i$. For every
  edge $e_{j}$ let $g_{j} \in A_i$ be a genotype with $g_{j}[i_j] =
  g_{j}[i_{j+1}] = 2$ and let $h_{j}$ and $h_{j}'$ be the explaining
  haplotypes for $g_j$. The above claim and the discussion
  about induced sets in Section~\ref{section:basics} imply the following two
  properties:  
  First, if $h_j$ and $h_j'$ resolve $g_j$ in columns~$i$ and~$i_j$ equally,
  then they
  resolve $g_j$ in columns~$i$ and $i_{j+1}$ equally if $e_{j}$ has
  weight 0 and unequally, if the weight is 1.  Second, if $h_{j}$ and
  $h_{j}'$ resolve $g_j$ in columns~$i$ and~$i_j$ unequally, they
  resolve $g_j$ in columns~$i$ and $i_{j+1}$ equally if $e_{j}$ has
  weight 1 and unequally, if the weight is 0.  Thus, when an edge
  $e_j$ has weight 1, the resolution alternates between the column
  pair~$i$ and $i_j$ and the column pair~$i$ and $i_{j+1}$. The
  resolution does not alter if the weight is 0.  Now assume that
  $g_{1}$ is resolved equally by $h_{1}$ and $h_{1}'$ in columns~$i$
  and $i_1$. From the above property follows that for every
  $j \in \{1,\dots,p\}$, the genotype $g_{j}$ is resolved equally by its haplotypes
  in columns~$i$ and $i_{j+1}$ if the path $i_1, e_{1},i_2 \dots i_{j},
  e_{j}, i_{j+1}$ has an even weight, and unequally otherwise. If we
  consider the whole odd-weight cycle, this fact yields a contradiction to
  the resolution of $g_{1}$ in~$i$ and $i_1$. The assumption
  that~$g_1$ is resolved unequally in columns~$i$ and~$i_1$ is also contradictory.

  \emph{If-part:} Let $A$ be a genotype matrix such that for every $i$ and $j$
  we have $\{01,10,11\} \nsubseteq \operatorname{ind}^A(i,j)$ and no
  resolution graph has an odd-weight cycle. We construct an
  explaining haplotype matrix $B$ for $A$ and show that it
  satisfies the three gamete property.

  \emph{Construction:} For genotypes without 2-entries, the explaining
  haplotypes are simply copies of them. The other genotypes (with
  2-entries) are partitioned by the sets $A_1, \dots, A_m$ and we
  treat every set $A_i$ and its graph~$G_i$ separately. First we
  transform~$G_i$ into a connected graph $G_i'$: From every component
  that does not contain~$i$, we select a vertex and connect it to~$i$
  by an even-weight edge. Since~$G_i$ does not contain odd-weight cycles,
  each vertex is connected to~$i$ by either only even-weight
  paths or only odd-weight paths in $G_i'$. For each genotype $g \in A_i$, we
  construct two explaining haplotypes $h$ and $h'$ as follows: For every
  column~$j$ with $g[j] = 2$, we set 
  $h[j] = 0$ and $h'[j] = 1$, if there is an even-weight path
  between~$i$ and~$j$ in $G_i'$ and $h[j] = 1$ and $h'[j] = 0$, otherwise. 
  The 0-entries and 1-entries are copied to both haplotypes.
  
  \emph{Three gamete property:} We are left to prove that  
  $\{01,10,11\} \nsubseteq \operatorname{ind}^B(k,l)$ holds for every column
  pair $k$ and $l$. If there is no genotype $g$ with 2-entries in $k$ and $l$, we
  simply have $\operatorname{ind}^B(k,l) = \operatorname{ind}^A(k,l)$
  and the three gamete property holds for that columns by assumption.
  Consider two columns $k$ and $l$ that have 2-entries in a common
  genotype. We distinguish whether~$11 \in \operatorname{ind}^A(k,l)$, 
  $\{01,10\} \subseteq \operatorname{ind}^A(k,l)$ or no of
  them holds and show that $\{01,10,11\} \nsubseteq \operatorname{ind}^B(k,l)$
  is true in all cases.
  \begin{Enumerate}
  \item Let $11 \in \operatorname{ind}^A(k,l)$, $g$ a genotype
  with $g[k] = g[l] = 2$ and $g \in A_i$ for a column $i$. We know that
  there is an edge $(0,\{k,l\})$ in $G_i'$ and, therefore, either all
  paths from~$i$ to~$k$ and all paths from~$i$ to~$l$ have an even
  weight or all paths from~$i$ to~$k$ and all paths from~$i$ to~$l$
  have an odd weight. The construction of the haplotypes implies
  that $g$ is resolved equally in $k$ and $l$. Thus, the induced set in columns
  $k$ and $l$ is extended only by the strings $00$ and $11$, which does not
  conflict with the three gamete property.
  \item Let $\{01,10\} \subseteq \operatorname{ind}^A(k,l)$,
  $g$ a genotype with $g[k] = g[l] = 2$ and~$i$ with $g \in A_i$. This
  implies that there is an edge $(1,\{k,l\})$ in $G'_i$ and,
  therefore, either all paths from~$i$ to~$k$ have an even weight and
  all paths from~$i$ to~$l$ have an odd weight or, conversely, all
  paths from~$i$ to~$k$ have an odd weight and all paths from~$i$
  to~$l$ have an even weight. Thus, $g$ is resolved unequally by it
  haplotypes in~$k$ and~$l$, which yields the additonal induced strings 
  $01$ and $10$.
  \item Assume that neither $11 \in \operatorname{ind}^A(k,l)$
  nor $\{01,10\} \subseteq \operatorname{ind}^A(k,l)$ holds. If the genotypes
  with 2-entries in columns~$k$ and~$l$ lie in the same set $A_i$, then
  they are all resolved equally or all resolved unequally in~$k$ 
  and~$l$, since their resolution depends on same graph~$G_i'$. If the
  genotypes with 2-entries in~$k$ and~$l$ are distributed among
  multiple sets $A_{i}$, the corresponding resolution graphs
  contain even-weight edges between~$k$ and~$l$ by definition. Similar to the
  first case, the construction assures that these genotypes are resolved
  equally in~$k$ and~$l$.
  \end{Enumerate}
  In all cases we proved that the resolutions of column pairs (and the resulting
  extension of induced sets) do not conflict with the three gamete property.
\end{proof}

\section{Reduction from PPH to the Bipartition Problem} \label{section:reduction}

We use our characterization from the last section to show that
\Lang{pph} can be reduced to the question of whether an undirected
graph is bipartite. Bipartite graphs are characterized by the absence of
an odd-length path and the formal $\Lang{bipartition}$
problem contains exactly the graphs with this property.

\begin{lemma}\label{lemma:reduction}
  $\Lang{pph}$ reduces to $\Lang{bipartition}$ via
  first-order-reductions.
\end{lemma}
\begin{proof}
  \emph{Construction:} The reduction procedure consists of three steps:
  For a given genotype matrix $A$, we first 
  apply the reduction from \Lang{pph} to \Lang{dpph} from
  \cite{EskinHK2003}, which is described in
  Section~\ref{section:basics}. This yields a genotype matrix $A'$.
  Then we construct a graph $G$ that is the disjoint union of
  all resolution graphs of column of $A'$. In the last step, we
  substitute each $0$ weighted edge in $G$ by a path of length 2
  and, finally, delete all edge weights. This yields the graph $G'$. 

  \emph{Correctness:} First, $A$ admits a perfect phylogeny exactly if $A'$
  admits a directed perfect phylogeny. Furthermore, we know from 
  Lemma~\ref{lemma:characterization} that $A'$ admits a directed
  perfect phylogeny if, and only if, $G$ does not contain an
  odd-weight cycle. The insertion of paths of length 2 transforms
  even-weight paths into even-length paths and odd-weight paths into
  odd-length paths. Therefore, $G$ does not contain an odd-weight cycle
  if, and only if, $G'$ is bipartite, which proves the lemma.
\end{proof}

Since $\Lang{bipartition} \in \Class{L}$~\cite{Reingold2008} and
$\Lang{pph}$ is $\Class{L}$-hard~\cite{ElberfeldT2008b}, the
reduction gives the last step to prove
Theorem~\ref{theorem:completeness} from the introduction.

\section{Conclusion} \label{section:conclusion}

In this paper we settled the main open problem
from~\cite{ElberfeldT2008b} and showed that perfect phylogeny
haplotyping is solvable in deterministic logarithmic space and,
therefore, is also $\Class{L}$-complete (a hardness proof can be found
in~\cite{ElberfeldT2008b}). We introduced a characterization of 
\Lang{pph} in terms of resolution graphs that represent resolutions of
2-entries in genotypes. We proved that the question of whether there
are resolutions that conflict with the existence of perfect
phylogenies is closely related to the bipartition problem for
undirected graphs. This yields a reduction from
\Lang{pph} to the bipartition problem, which can be seen as a 
conceptual easy and efficient approach to decide and
construct perfect phylogenies for genotype data.


\bibliographystyle{plain}

\begin{thebibliography}{1}
\bibitem{Bafnaetal2003}
V.~Bafna, D.~Gusfield, G.~Lancia, and S.~Yooseph.
\newblock Haplotyping as perfect phylogeny: A direct approach.
\newblock {\em Journal of Computational Biology}, 10(3--4):323--340, 2003.

\bibitem{DingFG2006}
Z.~Ding, V.~Filkov, and D.~Gusfield.
\newblock A linear-time algorithm for the perfect phylogeny haplotyping {(PPH)}
  problem.
\newblock {\em Journal of Computational Biology}, 13(2):522--553, 2006.

\bibitem{ElberfeldT2008b}
M.~Elberfeld and T.~Tantau.
\newblock Computational complexity of perfect-phylogeny-related haplotyping
  problems.
\newblock In {\em Proceedings of the International Symposium on Mathematical
  Foundations of Computer Science (MFCS 2008)}, volume 5162 of {\em Lecture
  Notes in Computer Science}, pages 299--310. Springer, 2008.

\bibitem{EskinHK2003}
E.~Eskin, E.~Halperin, and R.~M. Karp.
\newblock Efficient reconstruction of haplotype structure via perfect
  phylogeny.
\newblock {\em Journal of Bioinformatics and Computational Biology},
  1(1):1--20, 2003.

\bibitem{GrammNT2007a}
Jens Gramm, Arfst Nickelsen, and Till Tantau.
\newblock Fixed-parameter algorithms in phylogenetics.
\newblock {\em The Computer Journal}, 51(1):79--101, 2008.

\bibitem{Gusfield2002}
D.~Gusfield.
\newblock Haplotyping as perfect phylogeny: Conceptual framework and efficient
  solutions.
\newblock In {\em Proceedings of the Sixth Annual International Conference on
  Computational Molecular Biology (RECOMB 2002)}, pages 166--175. ACM Press,
  2002.

\bibitem{LiuZ2005}
Y.~Liu and C.-Q. Zhang.
\newblock A linear solution for haplotype perfect phylogeny problem.
\newblock In {\em Proceedings of the International Conference on Advances in
  Bioinformatics and its Applications}, pages 173--184. World Scientific, 2005.

\bibitem{Reingold2008}
O.~Reingold.
\newblock Undirected connectivity in log-space.
\newblock {\em Journal of the ACM}, 55(4):1--24, 2008.

\bibitem{SatyaM2006}
R.~Vijaya Satya and A.~Mukherjee.
\newblock An optimal algorithm for perfect phylogeny haplotyping.
\newblock {\em Journal of Computational Biology}, 13(4):897--928, 2006.

\end{thebibliography}

\end{document}